\documentclass[11pt]{article}
\usepackage{authblk}

\usepackage{thm-restate}
\usepackage{amsmath,amsthm,fullpage}
\usepackage{amssymb}
 \usepackage{breqn}
\usepackage{xcolor}
\usepackage{soul}
\usepackage{bbold}

\newtheorem{theorem}{Theorem}
\newtheorem{prop}[theorem]{Proposition}
\newtheorem{corollary}[theorem]{Corollary}

\newtheorem{proposition}[theorem]{Proposition}

\newcommand{\N}{\mathbb N}

\newcommand{\Zp}{\mathbb{F}_p}
\newcommand{\Zn}{\mathbb{Z}_n}

\newcommand{\Q}{\mathbb Q}

\newcommand{\Zz}{{\mathbb Z}(\zeta_n)}

\newcommand{\pid}{{\mathfrak p}}

\newcommand{\poly}{\mbox{{\sf poly}}}

\newcommand{\Pt}{\mbox{{\sf P}}}

\newcommand{\BPP}{\mbox{\sf{BPP}}}
\newcommand{\NC}{\mbox{{\sf NC}}}

\newcommand{\coNP}{\mbox{{\sf coNP}}}

\newcommand{\NP}{\mbox{{\sf NP}}}

\newcommand{\PSPACE}{\mbox{{\sf PSPACE}}}

\newcommand{\sct}{\mbox{{Sparse-CIT}}}
\newcommand{\skct}{\mbox{{Bounded-CIT}}}
\newcommand{\cct}{\mbox{{CIT}}}

\title{Cyclotomic Identity Testing and Applications}

\author[1]{Nikhil Balaji}
\author[2]{Sylvain Perifel}
\author[2]{Mahsa Shirmohammadi}
\author[3]{James Worrell}

\affil[1]{Department of Computer Science and Engineering, IIT Delhi, India}
\affil[2]{Universit\'e de Paris, CNRS, IRIF, F-75013 Paris, France}
\affil[3]{Department of Computer Science, University of Oxford, UK}

\begin{document}
\clearpage\thispagestyle{empty}
\maketitle

\begin{abstract}
We consider the cyclotomic identity testing (CIT) problem: given a 
polynomial~$f(x_1,\ldots,x_k)$, 
decide whether $f(\zeta_n^{e_1},\ldots,\zeta_n^{e_k})$ is zero, 
where $\zeta_n = e^{2\pi i/n}$ is a primitive complex $n$-th root of unity 
and $e_1,\ldots,e_k$ are integers, represented in binary.
When~$f$ is given by an algebraic circuit,  
we give a randomized 
polynomial-time algorithm for CIT
assuming  the generalised Riemann hypothesis (GRH), and  show that the  problem is  in {\coNP} unconditionally.
When $f$ is given by a circuit of polynomially bounded
degree, we give a randomized {\NC} algorithm. 
In case $f$ is a linear form
we show that the problem lies in~{\NC}.
Towards understanding when CIT can be solved in 
deterministic polynomial-time, we 
consider so-called diagonal depth-3 circuits, i.e., 
polynomials 
$f=\sum_{i=1}^m g_i^{d_i}$, where $g_i$ is a linear form
and $d_i$ a positive integer given in unary.  
We observe that a polynomial-time algorithm for CIT on this class would yield
a sub-exponential-time algorithm for polynomial identity testing.  However, 
assuming GRH, we show that 
if the linear forms~$g_i$ are all identical then CIT can be 
solved in polynomial time.
Finally, we use our results to give a new proof that equality of compressed 
strings, i.e., strings presented using context-free grammars, can be decided in 
randomized~{\NC}.
\end{abstract}


\section{Introduction}\label{sec:intro}
Identity testing in number fields is a fundamental problem in algorithmic algebra
that has been
studied in relation to solving systems of
polynomial equations~\cite{ge, koiran-hilbert} and  polynomial identity
testing~\cite{chen-kao}.
Among number
fields, cyclotomic fields, i.e., those generated by roots of unity,
play a central role.  The aim of this paper is
a comprehensive study of the computational complexity of identity testing in cyclotomic fields.

We consider \emph{cyclotomic identity testing problems}, where the input consists of 
a polynomial $f(x_1,\ldots,x_k)$ with integer coefficients, together with
integers $n,e_1,\ldots,e_k$, and the task
is to decide whether $f(\zeta_n^{e_1},\ldots,\zeta_n^{e_k})$ is zero for $\zeta_n=e^{2\pi i/n}$.
We consider four variants of this problem according to the representation of $f$: 
(i)~$f$ is given as an algebraic circuit;
(ii)~$f$ is given by a circuit of polynomially bounded syntactic degree;
(iii)~$f$ is a linear form; 
(iv)~$f=\sum_{i=1}^m g_i^{d_i}$ where each $g_i$ is a 
linear form and $d_i$ is an integer in unary.
Although $f$ is a multivariate polynomial,
since it is evaluated on powers of a common primitive $n$-th root of unity,
we formalise the above problems in terms of circuits whose input 
gates are labelled 
by powers of a single variable~$x$.  

Formally, for our purposes an algebraic circuit $C$ is a directed, acyclic graph with labelled vertices and edges.  
Vertices of in-degree zero are labelled in the set of monomials
    $\{x^e:e\in\mathbb{N}\}$ and the remaining vertices have labels in $\{+,\times\}$.
    Moreover the incoming edges to $+$-vertices have labels in $\mathbb{Z}$, that is, the $+$-gates compute integer-weighted sums.
    There is a unique vertex of out-degree zero which determines the output of the circuit, a univariate polynomial, in an obvious manner.
    The size of $C$ is the sum of the number of edges in the underlying graph and the bit-length of  all integer constants appearing in $C$.  The syntactic degree 
    of $C$ is defined inductively as follows: input gates have degree $1$, the degree of an addition gate is the maximum of the degrees of 
    its inputs, the degree of a multiplication gate is the sum of the degrees of its inputs, and the degree of $C$ is the degree 
    of the output gate.  Note that
    the syntactic degree of $C$ is \emph{not} an upper bound on the degree
    of the computed polynomial since we allow monomials as inputs.   Unless 
    otherwise stated, we assume that all integers are represented in binary.
       
The four main variants of the cyclotomic identity testing problem are as follows: 
 
 In the \emph{Cyclotomic Identity Testing (\cct)} problem the input is an algebraic circuit $C$ representing 
    a polynomial $f(x)$,
    together with an integer $n$, and the task is to determine whether $f(\zeta_n)=0$, where $\zeta_n = e^{2\pi i/n}$ is a 
    primitive complex $n$-th root of unity.
  
   The \emph{{\skct}} problem is defined exactly as the {\cct} problem, except that the input also includes
    an upper bound on 
 the syntactic degree of the circuit $C$ that is given in unary.  
 Thus in {\skct} the degree of the circuit is at most the length of the input. 

In the~\emph{Sparse-{\cct}} problem 
the circuit $C$ has syntactic degree~$1$. 
This is equivalent to giving the input polynomial 
$f$ in sparse representation, 
i.e., where $f=\sum_{i=1}^s a_ix^{k_i}$ 
is encoded as a list of pairs of integers 
$(a_1,k_1),\ldots,(a_s,k_s)$.

Finally we consider {\cct} in case $C$ is a diagonal circuit
(see~\cite{saxena-diagonal}), that is, the input polynomial has the form
$f=\sum_{i=1}^m g_i^{d_i}$ where each $g_i$ is in sparse representation and 
$d_1,\ldots,d_m$ are integers represented in unary.  

The representation of polynomials in the {\cct} problem 
can be exponentially more succinct than in the {\skct} problem,
since the syntactic degree of a circuit can be exponential 
in its size. Likewise the representation 
in the {\skct} problem can be exponentially more succinct than in the {\sct} problem,
since the former allows the number of monomials to be exponential in the circuit size.

The problem {\sct} was first studied by Plaisted~\cite{plaisted}, who gave a
randomised polynomial-time algorithm.   
Subsequently, two different deterministic
polynomial-time algorithms were given by Cheng et al.~\cite{cheng-tv,cheng-focs}.
A natural approach to decide zeroness of $f(\zeta_n)$ 
is to compute an approximation of {sufficient precision}.
However, given  existing separation bounds for algebraic numbers, the precision required to distinguish between zero and a non-zero value 
precludes a polynomial-time bound via this method.

The conclusion of~\cite{cheng-tv}
raises the question of the complexity of {\cct}.  The authors note
that this problem lies in the counting hierarchy (which lies between {\NP} and {\PSPACE}),
based on results of~\cite{ABKM}. 
Our first main result is that {\cct} can be placed in
{\BPP} (i.e., randomized polynomial time)  assuming GRH, and is in {\coNP} unconditionally.
The algorithm works by computing modulo a suitable prime ideal in the
ring of integers of the number field $\mathbb{Q}(\zeta_n)$.  
%

\begin{restatable}{theorem}{theocct} 
\label{theo:cct}
The {\cct} problem is in {\BPP} assuming GRH, and is in {\coNP} unconditionally.
\end{restatable}

Observe that the {\cct} problem is at least as hard as the 
Polynomial Identity Testing problem for circuits of unbounded degree, 
which is well-known to be {\Pt}-hard~\cite[Theorem 2.4.6, Theorem 2.6.3]{mittmann-thesis}.
As a result, algorithms for {\cct} are inherently sequential. 
We inspect two natural restrictions of {\cct} and show that they 
admit efficient parallel algorithms. 

The complexity class {\NC} formalises those polynomial-time-computable problems that are 
considered to be efficiently parallelizable.
Formally, a problem is in {\NC} if on instances of size $n$, it can be solved in $(\log n)^{O(1)}$
time using $n^{O(1)}$ processors in the PRAM model. Almost all natural problems in arithmetic~\cite{bch} and linear algebra are known to be in {\NC}~\cite{cook}. However, membership in {\NC} is open for 
GCD computation, modular powering, and primality testing.  (Several of these problems are, however, 
known to admit efficient \emph{randomized} {\NC} algorithms.)

We consider the {\skct} problem, in which the syntactic degree of the circuit is 
polynomially bounded, and 
give a randomized {\NC} procedure with two-sided errors.
Here we forsake the approach via finite arithmetic because computing powers in a 
finite field is not known to be in {\NC}.
Instead, we 
follow the identity testing method of Chen and Kao~\cite{chen-kao}: we pick a 
Galois conjugate of $f(\zeta_n)$ 
uniformly at random and determine
the zeroness of the conjugate by numerical computation.  
Thus we have:

\begin{restatable}{theorem}{theoskew} 
\label{theo:skew}
	The {\skct} problem is in randomized {\NC}.
\end{restatable}

Moving to the problem {\sct}, 
we revisit the approach of~\cite{cheng-tv}, who gave a polynomial-time
decision procedure.
Here we give a simpler reformulation of their method and, as a 
by-product, we observe that the problem
can be solved in {\NC}.

\begin{restatable}{theorem}{theosparse} 
\label{theo:sparse}
	The {\sct} problem is in {\NC}.
\end{restatable}

Theorems~\ref{theo:cct},~\ref{theo:skew}, and~\ref{theo:sparse} all take different approaches to the {\cct} problem: respectively
using finite arithmetic, numerical approximation, and multilinear algebra.
However it is interesting to note that all three approaches
involve computing a partial prime factorisation of the order of the root of unity (or some multiple thereof).

Intermediate between {\skct} and {\sct}, we consider CIT for 
diagonal circuits.  Here we observe that if CIT for diagonal circuits 
were solvable in 
polynomial time, then
polynomial identity testing (PIT) for algebraic circuits of size $s$ and 
degree $d$ would be sovlable in time polynomial in $s^{O(\sqrt{d})}$.  
However we have:

%
\begin{restatable}{theorem}{theodiamond}
\label{theo:diamond}
Assuming GRH, CIT can be solved in polynomial time on the class of 
polynomials of the form $f=\sum_{i=1}^m g^{d_i}$ with~$g$ in sparse representation and $d_i$ an integer
in unary.
\label{theo:diamond}
\end{restatable}

In terms of applications, we observe that cyclotomic identity testing can be used to 
obtain a new randomized {\NC} algorithm to decide equality of compressed strings, that is, strings presented 
by acyclic context-free grammars (see Section~\ref{sec:related} for previous work on this problem.)

\subsection{Related work}
\label{sec:related}

As discussed in~\cite{cheng-tv}, cyclotomic identity testing is related to the so-called \emph{torsion-point problem},
which asks whether a given multivariate polynomial has a zero in which all components are roots of unity~\cite{Rojas07}.
The univariate version of this problem is known to be {\NP}-hard~\cite{plaisted}.

Theorem~\ref{theo:cct} is a generalisation of the problem of testing equality of straight-line programs over the integers as studied by Sch{\"o}nhage~\cite{schonhage} (see also Allender et al.~\cite{ABKM}) to 
cyclotomic number fields. Testing zeroness of expressions involving real roots of rational 
numbers is considered in~\cite{blomer}.

Lenstra~\cite{LL} and Kaltofen and Koiran~\cite{kaltofen-koiran} gave 
polynomial-time algorithms for testing zeroness of sparse univariate and 
multivariate polynomials respectively on algebraic numbers of degree 
polynomially bounded in the problem instance (whereas {\sct} features 
roots of unity whose degree can be exponential in the size of 
the problem instance).

There has been extensive work on the problem of testing equality of compressed strings, including 
Hirschfield et al.~\cite{HJM} ($O(n^4)$ time), Plandowski~\cite{plandowski} and Melhorn et al.~\cite{MSU} ($O(n^3)$ time),
and Je\.{z}~\cite{jez} ($O(n^2)$ time).  Note that the quadratic running time in the latter is in a RAM model where 
arithmetic on integers (the binary encoding of a position in the uncompressed string 
fits into a single machine word) can be performed in a single time step. 

K\"{o}nig and Lohrey~\cite{konig-lohrey}  show that the problem admits a randomised {\NC} algorithm by
reduction to the identity testing problem for univariate polynomials given as so-called powerful skew circuits.
The main contribution of~\cite{konig-lohrey} is a randomised {\sf NC} algorithm for the latter problem.
Following the identity testing
 technique of Agrawal and Biswas~\cite{agrawal-biswas},
their algorithm works,
by computing the value of the circuit
modulo a randomly chosen polynomial $p(x)$.  In order to perform this computation in {\NC} they rely on the result of Fich and Tompa~\cite{fich-tompa} that
computing $x^m \bmod p(x)$ for large powers $m$ can be done in {\NC} (assuming $p$ is given in dense representation).  
By contrast, we observe that the same identity testing problem 
can be solved by 
numerically evaluating a polynomial at a randomly chosen conjugate of a root of unity $\zeta_n$ of sufficiently high order.  
To obtain an {\NC} bound we rely on the fact that it is straightforward to compute powers of~$\zeta_n$. We also observe that our technique yields a randomised sequential algorithm that runs in $\widetilde{O}(n^2)$ time in the standard Turing machine model.

\section{Preliminaries}\label{sec:pre}
We give some background results on arithmetic and cyclotomic fields.
Throughout the paper
we use $\log{x}$ to denote $\log_2{x}$. 

\begin{restatable}{proposition}{propZncondSn} 
\label{prop:ZncondSn}
Fix $m\in\mathbb{N}$ and
consider drawing an element $k$ uniformly at random from the set $\{1,\ldots,m-1\}$.
Let $A$ be the event that $k$ and $m$ are coprime
and let $B$ be the event that $k$ and $m$ share no common prime divisor $p< 10\log m$.  
Then $\mathrm{Pr}(A\mid B)>\frac{9}{10}$ for $m$ sufficiently large.
\end{restatable}

\begin{proof}
Write $m=p_1^{e_1}\cdots p_r^{e_r}$, 
where $p_1,\ldots,p_r$ are distinct primes and $e_1,\ldots,e_r \geq 1$.
For $i=1,\ldots,r$, let $E_i$ be the event that $p_i$ does not divide the sampled number $k$.
Then the collection of events~$E_i$ is	mutually independent, $A=\bigcap_{i=1}^r E_i$, and
$B=\bigcap_{i:p_i<10\log m} E_i$.  Thus
\begin{eqnarray*} 
\mathrm{Pr}(A \mid B) &= & \frac{\mathrm{Pr}(A)}{\mathrm{Pr}(B)} \; = \prod_{i:p_i \geq 10\log m} \mathrm{Pr}(E_i)\\
                      &=& \prod_{i:p_i \geq 10\log m} \left(1-\frac{1}{p_i}\right) \\
      	     	     	&\geq & \left(1-\frac{1}{10\log m}\right)^{\log m} \, .
\end{eqnarray*}
Since the expression above converges to $e^{-0.1} > 0.9$ as $m$ tends to infinity, 
for sufficiently large $m$ we have $\mathrm{Pr}(A \mid B) > \frac{9}{10}$.
\end{proof}

The following estimates on the density of primes in an arithmetic progression can be found in~\cite[Chapter 20, page 125]{davenport} and~\cite[Corollary 18.8]{IKbook} respectively.

\begin{theorem}\label{theo-prime}
Given $a\in \Zn^*$, 
write $\pi_{n,a}(x)$ for the number of primes less than $x$ that are congruent to~$a$ modulo~$n$.
Then under GRH, there is an absolute constant $C>0$ such that
$$\pi_{n,a}(x) \geq  \frac{x}{\varphi(n) \, \log x} - C\sqrt{x}\log(x).$$
Unconditionally, 
there exist effective absolute positive 
constants $C_1$ and $C_2$ such that
for all $n < C_1x^{C_1}$,
$$ \pi_{n,a}(x) \geq \frac{C_2 x}{\varphi(n) \sqrt{n}\, \log(x) }\, . $$
\end{theorem}

Recall that $\alpha\in\mathbb{C}$ is 
an \emph{algebraic number} if it is the root of a non-zero polynomial in $\mathbb{Q}[x]$.  The \emph{minimal polynomial} of $\alpha$ is the unique monic polynomial in $\mathbb{Q}[x]$ having $\alpha$
as a root.  If the minimal polynomial has integer coefficients then we say that $\alpha$ is an 
\emph{algebraic integer}.  The degree and height of an algebraic integer are respectively the degree and the maximum absolute value of a coefficient in its minimal polynomial.

Fix $n\in\mathbb{N}$ and write $\mathbb{Q}(\zeta_n)$ for the
the field obtained by adjoining a primitive complex $n$-th root of
unity~$\zeta_n=e^{\frac{2\pi i}{n}}$ to~$\mathbb{Q}$.
The $n$-th \emph{cyclotomic polynomial} is the minimal polynomial of $\zeta_n$.
It is well known that $\Phi_n$ has degree $\varphi(n)$,
where $\varphi$ is the Euler totient function. 

It is well known that the sub-ring 
of $\mathbb{Q}(\zeta_n)$ comprised of algebraic integers, called the sub-ring of cyclotomic integers,  is the ring $\mathbb{Z}[\zeta_n]$ that is generated over $\mathbb{Z}$ by $\zeta_n$.

Recall that 
the group $\mbox{Gal}(\Q(\zeta_n)/\Q)$ of automorphisms of
$\Q(\zeta_n)$ is isomorphic to the multiplicative group~ $\Zn^*$ 
of integers mod $n$.
For each $k\in \Zn^*$, the corresponding automorphism $\sigma$ in $\mbox{Gal}(\Q(\zeta_n)/\Q)$ is defined by $\sigma(\zeta_n)=\zeta_n^k$.

The image of $\alpha \in \Q(\zeta_n)$ under an automorphism of 
$\Q(\zeta_n)$ is called a \emph{Galois conjugate} of $\alpha$.
The \emph{norm} of $\alpha$ is defined by
\[
N_{\Q(\zeta_n)/\Q}(\alpha) := \prod_{\sigma \in \mbox{Gal}(\Q(\zeta_n)/\Q)} \sigma(\alpha) \, 
\]
For short, we will write $N(\alpha)$ for $N_{\Q(\zeta_n)/\Q}(\alpha)$,
i.e., the underlying field will be understood from the context.
Recall that the norm of a cyclotomic integer lies in $\mathbb{Z}$.

\section{A Randomised Polynomial-time Algorithm for \cct}\label{sec:bbp}
In this section  we give a randomised polynomial-time algorithm for the {\cct} problem, assuming GRH.
Furthermore, we show that the  problem is in {\coNP} unconditionally.
The idea is to work in a finite field, obtained by quotienting 
the ring of cyclotomic integers by a suitable prime ideal.


Throughout this section, 
we say that the cyclotomic integer $f(\zeta_n)$ has \emph{a description of 
   size~$s \in \mathbb{N}$} if it is computed by 
 a circuit $C$  
 where 
 $s$ is  the sum of the size of $C$ and the bit-length of $n$.


\begin{restatable}{proposition}{propbound} 
\label{prop:bound}
Let $\alpha \in \Zz$ be a cyclotomic integer  with  a description
of size~$s$.  Then we have~$|N(\alpha)| \leq 2^{2^{2s}}$.
\end{restatable}

\begin{proof}
Write $\alpha = \sum_{j=0}^{n-1} a_j \zeta_n^j$, where $a_0,\ldots,a_{n-1} \in \mathbb{Z}$
and let
$H := \sum_{0\leq i \leq n-1} |a_i|$. Since $\alpha$ is computed by a circuit of size~$s$, 
by an easy induction on~$s$
we have $H \leq 2^{2^s}$. Then:
\begin{align*}
N(\alpha)\, & =\, N\left(\sum_{j=0}^{n-1} a_j \zeta_n^j\right) \\
\,& \, = \prod_{\sigma \in \mbox{Gal}(\Q(\zeta_n)/\Q)} \sigma\left(\sum_{j=0}^{n-1} a_j \zeta_n^j\right) \\
&\, = \prod_{\ell \in \mathbb{Z}^*_{n}}\left(\sum_{j=0}^{n-1} a_j \zeta_n^{j\ell}\right). 
\end{align*}
Since $\left| \sum_{j=0}^{n-1} a_j \zeta_n^{j\ell} \right | \leq \sum_{i=0}^{n-1} |a_i| = H $ for all $\ell$,   we have
\[
|N(\alpha)|\,  \, \leq \prod_{\ell \in \mathbb{Z}^*_{n}} H 
\, \leq (2^{2^s})^{n} 
\, \leq (2^{2^s})^{2^s} 
 =  2^{2^{2s}}.\]
\end{proof}

\begin{theorem}
Let $p\in\mathbb{Z}$ be a prime such that  $\Zp$ contains a
primitive $n$-th root of unity~$\omega_n$.  Given $g(x) \in  
\mathbb{Z}[x]$, denoting by $\bar{g}\in \Zp[x]$ the reduction of 
$g$ modulo $p$, we have 
\begin{enumerate}
  \item  if $g(\zeta_n)=0$ then $\bar{g}(\omega_n)=0$, and
  \item if $\bar{g}(\omega_n)=0$ then $p                                                                               
\mid N(g(\zeta_n))$.
\end{enumerate}
\label{thm:soundness}
\end{theorem}
\begin{proof}
Define a ring homomorphism $\mathrm{ev}:\mathbb{Z}[x]\rightarrow \Zp$ by $\mathrm{ev}(g)=\bar{g}(\omega_n)$.
For $d<n$, since $\Phi_d \mid x^d-1$ and $\mathrm{ev}(x^d-1)\neq 0$, we have $\mathrm{ev}(\Phi_d)\neq 0$.
Since also $x^n-1=\prod_{d\mid n} \Phi_d$, we have $\mathrm{ev}(\Phi_n)=0$.  
It follows that $\mathrm{ev}$ factors through $\Zz$ via a homomorpishm $\mathrm{ev}':\Zz\rightarrow\Zp$ 
given by $\mathrm{ev}'(g(\zeta_n))=\bar{g}(\omega_n)$ for $g\in \mathbb{Z}[x]$. 

For Item 1, if~$g(\zeta_n)=0$ then~$\bar{g}(\omega_n)=\mathrm{ev}'(g(\zeta_n))=0$.

For Item 2, observe that the kernel of $\mathrm{ev}'$ is a prime ideal~$\pid$ in $\Zz$ satisfying $\pid \cap \mathbb{Z} = p\mathbb{Z}$.
Hence if $\bar{g}(\omega_n)=0$ then 
$g(\zeta_n) \in \pid$ and so $p \mid                                                                                              
N(g(\zeta_n))$.
\end{proof}

Theorem~\ref{thm:soundness} suggests a natural test for  \cct :
evaluate the circuit in a finite field~$\Zp$ that contains a primitive
$n$-th root of unity.  Since the multiplicative group
$\Zp^*$ is cyclic, it is clear that $\Zp^*$ contains a
primitive $n$-th root of unity just in case $n \mid (p-1)$, i.e., $p
\equiv 1 \bmod{n}$. 

\begin{proposition}\label{prop:comp}
Let $\alpha\in \Zz$ be a non-zero cyclotomic integer whose description has size
at most $s$.  Suppose that $p$ is chosen uniformly at random from $\{ q \in \mathbb{N} : q \leq 2^{5s} \text{ and }{q\equiv 1}\bmod n\}$.
Assuming GRH, (i)~$p$ is prime with probability at least $\frac{1}{6s}$, and
(ii)~given that $p$ is prime, the probability that it divides $N(\alpha)$ is at most $2^{-s}$.
\end{proposition}
\begin{proof}
For (i), we note that 
by Theorem~\ref{theo-prime}, the probability that $p$ is prime is
at most 
\[
\frac{\pi_{n,1}(2^{5s})}{2^{5s}/n} \geq  \frac{n}{\varphi(n)\log(2^{5s})}
- \frac{Cn\log(2^{5s})}{\sqrt{2^{5s}}} \\
 \geq  \frac{1}{5s}  - \frac{C2^s5s}{2^{2s}} \, ,
\]
where $C$ is the absolute constant mentioned in the theorem.   But the above
is at least $\frac{1}{6s}$ for $s$ sufficiently large.

For (ii),
by Proposition~\ref{prop:bound}
the norm of $\alpha$  has absolute value  at most $2^{2^{2s}}$,
and hence $N(\alpha)$ has at most $2^{2s}$ distinct prime factors.
Then, for $s$ sufficiently large, the probability that $p$ divides $N(\alpha)$ given that $p$ is prime is at most $\frac{6sn2^{2s}}{2^{5s}} \leq 2^{-s}$.
\end{proof}



\begin{figure}[t]
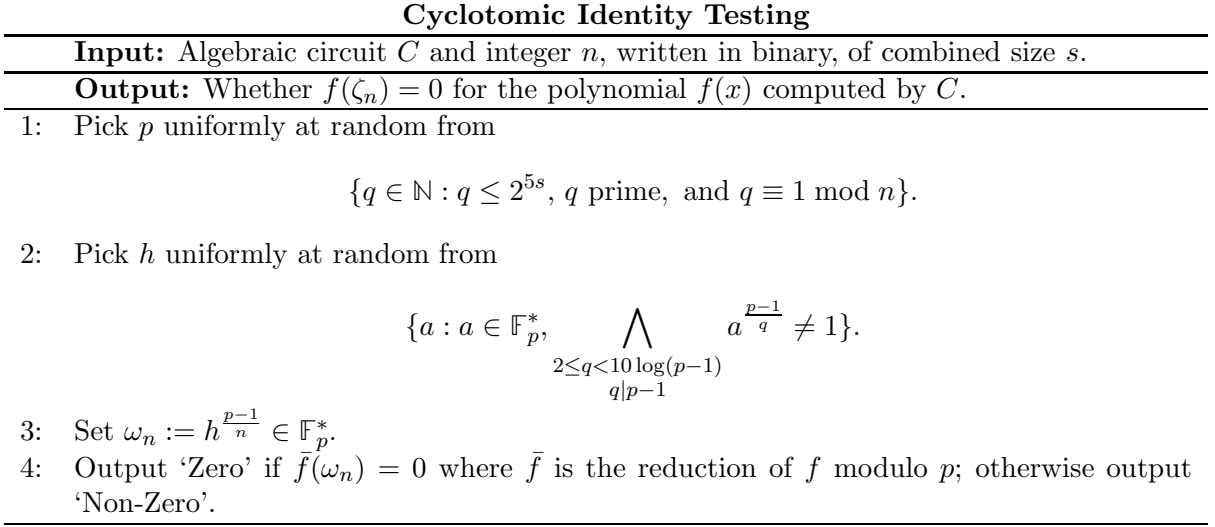

\begin{center}
\begin{tabular}{rp{.9\textwidth}}
 \hline
    \multicolumn{2}{c}{ \bf   Cyclotomic Identity Testing }  \\
    \hline
 & {\bf Input:} Algebraic circuit~$C$ and integer~$n$, written in binary, of combined size~$s$.
\\\hline
& {\bf Output:} Whether $f(\zeta_n)=0$ for the polynomial~$f(x)$ computed by~$C$.\\
  \hline
  1:& Pick $p$ uniformly at random from \[\{ q \in \mathbb{N}: q \leq 2^{5s},\, q \text{ prime},\text{ and }{q \equiv 1} \bmod n\}.\]\\[-.5cm]
  2:& Pick $h$ uniformly at random  from \[\{ a : a \in \Zp^*, \bigwedge_{\substack{2\leq q<10\log(p-1)\\ {q|p-1}}} a^{\frac{p-1}{q}} \neq 1\}.\]\\[-.5cm]
    3:& Set $\omega_n:=h^{\frac{p-1}{n}} \in \Zp^*$.\\
  4:& Output `Zero' if $\bar{f}(\omega_n)= 0$ where $\bar{f}$ is the reduction of~$f$ modulo~$p$; otherwise output
 `Non-Zero'.
\\
   \hline
\end{tabular}
\caption{Algorithm for the CIT problem}
\label{fig:algorithmBPP}
\end{center}
\vspace{-.5cm}
\end{figure}

A  straightforward application 
of Proposition~\ref{prop:ZncondSn} gives the following proposition,  
enabling us to find primitive $n$-th roots of unity in $\Zp$ in 
case ${p\equiv 1} \bmod n$.
\begin{proposition}
For a prime $p$, let
$h$ be chosen uniformly at random from the set 
\[ \Bigg\{a \in \Zp^*  :  \bigwedge_{\substack{2\leq q<10\log(p-1)\\ {q|p-1}}} a^{\frac{p-1}{q}} \neq 1\Bigg\} \, .\]
Then $h$ is a primitive root of  $\Zp^*$ with probability at least 0.9.
\label{prop:prim}
\end{proposition}
\begin{proof}
Fix a primitive root $g \in \Zp^*$. 
For $a$ distributed uniformly at random over $\Zp^*$,
we have that $\log_g a$ (the discrete logarithm of $a$ in $\Zp^*$) is distributed uniformly at random over $\{0,\ldots,p-2\}$.
Moreover, for every divisor $q$ of $p-1$, $q$ divides  
$\log_g a$ if and only if ${a^{\frac{p-1}{q}} = 1}\bmod p$.  It follows that 
for $h$ as in the statement of the proposition,
$\log_g h$ is distributed uniformly at random among those elements
in $\{2,\ldots,p-2\}$ that do not share a divisor less than $10 \log (p-1)$
with~$p-1$.  Applying Proposition~\ref{prop:ZncondSn} we have that $\log_g h$ is coprime with~$p-1$ with 
probability at least $0.9$.  But $\log_g h$ is coprime with $p-1$ if and only if $h$ is itself 
a primitive root of~$\Zp^*$.  
\end{proof}

We are now in a position to prove the main result of this section: 

\theocct*
\begin{proof}
Figure~\ref{fig:algorithmBPP} presents a Monte Carlo randomized
algorithm for the  \cct~problem. 
The argument for the correctness
of the algorithm is as follows.  Let $p$ be a prime such that~$p\equiv 1
\bmod{n}$, as chosen in Line~$1$.  

It follows from Proposition~\ref{prop:prim} that 
with probability at least $0.9$,
the element $h\in \Zp^*$ that is selected in Line~$2$ is a primitive root of $\Zp^*$.
Now let us bound the error of the algorithm under the assumption that $h$ is indeed a  primitive root of $\Zp^*$.  
Note that in this case we have that $\omega_n$, as chosen in Line 3, is a primitive $n$-th root of unity
in the field $\Zp$.  We consider two cases.
First, suppose that $f(\zeta_n)=0$; then by
Theorem~\ref{thm:soundness} we have $\bar{f}(\omega_n)=0$, and hence the
output is `Zero'.  Second, suppose that $f(\zeta_n)\neq 0$. Then by Theorem~\ref{thm:soundness}
the output will be `Non-Zero' provided that $p$ is does not divide
$N(f(\zeta_n))$.  But by Proposition~\ref{prop:comp}(ii)
the probability that $p$ does not divide $N(f(\zeta_n))$ is at least
$1-2^{-s}$.
Thus,  in total, the probability that the algorithm gives the wrong output is at most~$0.1 + 2^{-s}$.

It is clear that the algorithm runs in polynomial time.  In particular, in Line 1, by Proposition~\ref{prop:comp}(i) we can choose a prime uniformly 
at random from the set
$\left\{ q\in\N: q \leq 2^{5s},\, {q \equiv 1} \bmod n\right\}$
by random sampling with $O(s)$ repetitions with a small constant failure probability.


It remains to show that the CIT problem lies in {\coNP}.  The idea is to modify the algorithm in Figure~\ref{fig:algorithmBPP}, replacing randomisation with guessing.
Suppose $f(\zeta_n)\neq 0$.
The unconditional lower bound 
in Theorem~\ref{theo-prime} shows that
$\pi_{n,1}(2^{s'}) > 2^{2s}$ for $s$ sufficiently large and $s'>\max(4s, s/C_1)$.  It follows that there 
exists a prime $p \leq 2^{s'}$ that does not
divide $N(f(\zeta_n))$
such that 
${p\equiv 1}\bmod n$.  The polynomial certificate
of non-zeroness of $f(\zeta_n)$ then 
comprises, the above prime $p$, a list 
of the prime factors of $p-1$, and
an element $h \in \Zp^*$ such that 
$h^{\frac{p-1}{q}} \neq 1$ for all prime factors $q$ of 
$p-1$.  Such an $h$ is a generator  
of $\Zp^*$ and so $\omega_n:=h^{\frac{p-1}{n}}$
is a primitive $n$-th root of unity.  We then 
have that $\overline{f}(\omega_n)\neq 0$.  
On the other hand, as noted above, for any prime $p$ and primitive $n$-th root of unity
$\omega_n \in \Zp^*$, if $f(\zeta_n)=0$
then $\overline{f}(\omega_n)=0$.
%
\end{proof} 

\begin{corollary}
Assuming GRH, there is a Monte Carlo randomised algorithm that solves {\cct} for cyclotomic integer of description size $s$, using $O(s^2)$ random bits and $O(s^3)$ arithmetic operations.
\end{corollary}
\begin{proof}
For primality testing, we use the Miller-Rabin~\cite{schoof} test, which can be implemented using $O(s^3)$ arithmetic operations in $\Zp$ and $O(s)$ random bits with error at most $2^{-s}$. Since this step is repeated $O(s)$ times, Line 1 of the algorithm can be implemented using $O(s^4)$ arithmetic operations and $O(s^3)$ random bits. Line $2$ can be implemented using $O(s)$ random bits and $O(s^2)$ arithmetic operations in $\mathbb{F}_p$. Finally computing $\bar{f}(\omega_n)$ can be implemented using $O(s)$ arithmetic operations in $\mathbb{F}_p$.
\end{proof}

The following example illustrates that 
if the element $h$ chosen in Line~$2$ of the algorithm in Figure~\ref{fig:algorithmBPP} is not a primitive element of  $\Zp^*$ then there might be
two-sided errors in the output.
 Indeed, consider the behaviour of the 
 algorithm on input $n=12$
 and the two 
 cyclotomic polynomials 
 \[\Phi_{12}(x)=x^4-x^2+1, \text{and } \Phi_6(x)=x^2-x+1.\]
The output should be `Zero' for $\Phi_{12}$ and 'Non-Zero'
for $\Phi_6$.  However
suppose the algorithm 
chooses $p=13$ in Line~$1$ and $h=4$ in Line~$2$.
(Note that $h$ is \emph{not} a primitive element of $\mathbb{F}_{p}^*$ since it has order $6$). In this case
$\omega_n$ would be set to $4$ in Line~3.  But
$\Phi_{12}(4)\not \equiv 0 \pmod{13}$ and $\Phi_6(4)\equiv  0 \pmod{13}$.

\section{A Randomised NC Algorithm for \skct}
\label{sec:RandNC}
In this section we give a randomized {\NC} algorithm for the {\skct} problem.
The algorithm is shown in Figure~\ref{fig:algorithmSkew}.  The input is an
algebraic circuit~$C$ 
with a unary upper bound on its
  syntactic degree (but with high-powered inputs), together with an integer~$n$ written in binary.  
  The desired output is whether or not $f(\zeta_n)=0$ for $f(x)$ the polynomial represented by $C$.
  
  Intuitively, Line~1 of the algorithm attempts to select a Galois 
  conjugate $\zeta_n^a$ of $\zeta_n$ uniformly at random.  
   Here, to avoid having to check whether $\gcd(a,n)=1$---which is not known to be 
   in {\NC}---we only have a Galois conjugate with high probability.
  Line~2 computes
  a numerical approximation $\alpha$ of $f(\zeta_n^a)$ to precision $2^{-4s-1}$.  Finally, Line~3
  outputs `Zero' if and only if $\alpha$ has absolute value at most $2^{-4s-1}$.
 
\paragraph{\bf Numerical approximation.}
  We now explain how
  the numerical approximation
  $\alpha$ in Line~2 can be computed in {\NC} with the desired precision.  
  
  Write $\varepsilon:=2^{-s^2-5s-1}$.  
  Taking $O(s^2)$ terms of the power series involved in Machin's formula~\cite{machin} for $\pi$ we obtain 
    $\widetilde{\pi} \in \mathbb{Q}$ such that 
  $|\pi-\widetilde{\pi}|<\varepsilon/3$.    Likewise, 
  for each $n$-th root of unity $\zeta_n^\ell$, by 
  taking $O(s^2)$ terms in the power series expansion
  for $e^{2\widetilde{\pi}i\ell/n}$, we obtain
  $\widetilde{\zeta_n^\ell} \in \mathbb{Q}(i)$ such that 
  $|\widetilde{\zeta_n^\ell}-e^{2\widetilde{\pi}i\ell/n}| < \varepsilon/3$.  
  We then have that, for all $0\leq \ell<n$,
  \begin{eqnarray*}
  |\widetilde{\zeta_n^\ell}-\zeta_n^\ell| &\leq & |\widetilde{\zeta_n^\ell}-e^{2\widetilde{\pi}i\ell/n}|+
                                                  |e^{2\widetilde{\pi}i\ell/n}-e^{2{\pi}i\ell/n}|\\
                                &\leq & \varepsilon/3 + |1 - e^{2(\widetilde{\pi}-\pi)i\ell/n}| \\
                                &\leq & \varepsilon/3 + 2|\widetilde{\pi}-\pi| \\
                                & \leq & \varepsilon \, .
  \end{eqnarray*}

\begin{figure}[t]
\begin{center}
\begin{tabular}{rp{.9\textwidth}}
 \hline
    \multicolumn{2}{c}{ \bf   Bounded Cyclotomic Identity Testing }  \\
    \hline
  & {\bf Input:} Algebraic circuit~$C$ with a unary upper bound on its
 syntactic degree, and integer~$n$ written in binary, of combined size~$s$.
\\\hline
& {\bf Output:} Whether $f(\zeta_n)=0$ for the polynomial~$f(x)$ computed by~$C$.\\
 \hline
  1:& Pick uniformly at random $a\in \{1,\ldots,n-1\}$ such that $a$ and $n$ have no common divisor less than
  $10\log n$.
 \\
   2:&  Compute $\alpha \in \mathbb{Q}(i)$ such that $|f(\zeta_n^a) - \alpha| < 2^{-4s-1}$.\\
   3: & Output `Zero' if  $|\alpha| <2^{-4s-1}$, otherwise output `Non-Zero'.
\\
   \hline
\end{tabular}
\caption{Algorithm for the Bounded-CIT problem}
\label{fig:algorithmSkew}
\end{center}
\end{figure}

  Now $f(x)=g(x^{e_1},\ldots,x^{e_k})$, where $e_1,\ldots,e_k$ are positive integers and
  $g$ is a multivariate polynomial represented  by a circuit 
  of size $s$ and syntactic degree $d \leq s$.  We define
  \[ \alpha:=g(\widetilde{\zeta_n^{ae_1}},\ldots,\widetilde{\zeta_n^{ae_k}}) \in \mathbb{Q}(i) \]
  In other words, $\alpha$ is obtained by evaluating the circuit $C$ on inputs 
  $\widetilde{\zeta_n^{ae_1}},\ldots,\widetilde{\zeta_n^{ae_k}}$.  Now each monomial in $g$ has total degree at most $d$ and the sum of the aboslute values of the coefficients of $g$ is at most $2^{sd}$.  Thus we have, by definition of $\varepsilon$,
    \begin{eqnarray*}
  |f(\zeta_n^a)-\alpha|&=&|g(\zeta_n^{ae_1},\ldots,\zeta_n^{ae_k})-g(\widetilde{\zeta_n^{ae_1}},\ldots,\widetilde{\zeta_n^{ae_k}})|\\
                       &\leq & 2^{sd} |1-(1+\varepsilon)^d| \\ &\leq &
                       2^{sd} \varepsilon 2^d \, \leq
                        2^{-4s-1} \, .
  \end{eqnarray*}
  
  By construction, the approximants $\widetilde{\zeta_n^{ae_1}},\ldots,\widetilde{\zeta_n^{ae_k}}$ are
  represented by arithmetic circuits of size and degree polynomial in $s$ that can moreover 
  be computed in space
  logarithmic in~$s$.  Composing these circuits with $C$ gives an arithmetic circuit 
  for $\alpha$ that has size and degree polynomial in $s$.
  We now use the classical result of Valiant \emph{et al.}~\cite{VSBR} 
  that given an arithmetic circuit of degree $d$ 
and size $\sigma$ 
one can construct (in logarithmic space~\cite{ajmv}) an equivalent arithmetic circuit of depth 
$O(\log(d)\log(\sigma))$ and size 
polynomial in $\log d$ and $\sigma$.   
Every bit of the numbers produced at each gate of the resulting compressed-depth circuit can be 
computed by Boolean {\NC} circuits 
of size at most $O(\sigma\log{\sigma})$~\cite{reif-tate}.  Applying
this transformation to the algebraic circuit for $\alpha$
results in a Boolean circuit for (the bits of) $\alpha$ of size polynomial in $s$ and depth
polynomial in $\log(s)$ that is moreover computable in space $\log s$.  This shows that 
$\alpha$ is computable in {\NC}.

\paragraph{\bf Correctness.}\
The probabilistic correctness of the algorithm relies on the
following well-known result:
\begin{restatable}{proposition}{chenkao}[Chen and Kao~\cite{chen-kao} and Bl\"{o}mer~\cite{blomer}]\label{prop:blomer}
	Let $\alpha$ be a non-zero algebraic integer where  the absolute value of all its conjugates
	is at most $2^B$. For all $b \in \N$, a random
	conjugate $\alpha'$ of $\alpha$ satisfies $|\alpha'| \leq 2^{-b}$, with probability at most $B/(b + B)$.
\end{restatable}
 
 \begin{proof}
For  the algebraic integer~$\alpha$, let   $\alpha_0 = \alpha, \alpha_1, \dots, \alpha_{d-1}$ be the conjugates.
	Let $k$ be the number of conjugates that are at most $2^{-b}$ in absolute value. Recall that $|\prod_{i=0}^d \alpha_i|$ is 
	the absolute value of the constant term of the minimal polynomial of $\alpha$. Since the minimal polynomial by definition is integral, the product $|\prod_{i=0}^d \alpha_i|$ is at least $1$.
	Together with the upper bound $2^B$ on $|\alpha_i|$ we get
	\[
	1 \leq |\prod_{i=0}^d \alpha_i| \leq (2^{B})^{d-k} (2^{-b})^k
	\]
	
	This implies that 
	\begin{eqnarray*}
		(2^{B})^{d-k} (2^{-b})^{k} &\geq& 1 \\
		2^{dB} 2^{-k(B+b)} &\geq& 1 \\
		dB - k(B+b) &\geq& 0 \\
		\frac{k}{d} &\leq& \frac{B}{B+b}
	\end{eqnarray*}
\end{proof}
 
Observe that the constants appearing in $f$ have magnitude at most $2^s$ and that $f$ has degree less than $2^s$.
It follows that $|f(\zeta_n^\ell)|\leq 2^{2s}$ for all $\ell$.
Using Proposition~\ref{prop:blomer}  with $B = 2s$ and $b = 4s$, whenever $f(\zeta_n)$ is non-zero,  a  random conjugate
$f(\zeta_n^a)$ of~$f(\zeta_n)$  has 
absolute value larger than $2^{-4s}$ with probability at least~$\frac{2}{3}$. 

In Line~1, the algorithm takes a polynomial number of samples from $\{1,\ldots,n-1\}$ uniformly at random (in parallel) and returns any~$a$ 
such that $a$ and $n$ have no common divisors less than $10 \log{n}$.  By Proposition~\ref{prop:ZncondSn},  
we have that $a$ is coprime with $n$, and hence
$\zeta_n^a$ is a conjugate of $\zeta_n$, with probability at least~$\frac{9}{10}$.

To estimate the error probability we consider two cases.  First, suppose
that $f(\zeta_n)\neq 0$.  Then with probability at least $\frac{9}{10} \cdot \frac{2}{3} = \frac{3}{5}$ we have
that $|f(\zeta_n^a)|> 2^{-4s}$ and hence $|\alpha| > 2^{-4s-1}$.
Second, suppose that $f(\zeta_n)=0$.  Then with probability at least $\frac{2}{3}$ we have 
$f(\zeta_n^a)=0$ and hence $|\alpha| < 2^{-4s-1}$.
Thus the error probability is at most $\frac{2}{5}$.  Finally we obtain:

\theoskew*

\subsection{Compressed Words and Powerful Skew Circuits}
An algebraic circuit computing a univariate polynomial is said to be a 
\emph{powerful skew circuit} if at least one input of every multiplication gate is a leaf.
Here the word \emph{powerful} reflects our convention that leaves can be labelled with
monomials $x^m$, where $m$ is given in binary.
The main motivation for studying this identity testing problem
is that there is an {\NC} reduction of the equivalence testing problem for compressed strings to identity testing for powerful skew circuits~\cite{konig-lohrey}.
Briefly, a compressed word is one that is given by an acyclic context-free grammar in which each non-terminal occurs on the left-hand side of 
exactly one production.  Such a grammar produces a single word, whose length can be exponential in the number of non-terminals and productions.
See~\cite{konig-lohrey} for more details.

In this section 
we give an alternative randomised $\NC$ algorithm for PIT on powerful skew circuits,
employing the same random conjugate technique used to solve the {\skct} problem. 
Since the syntactic degree of a powerful skew circuit is at most the number of gates we can use our Algorithm in~Figure~\ref{fig:algorithmSkew} to decide
PIT over the class of powerful skew circuits: we simply 
pick a root of unity $\zeta_n$ with~$n$ higher than the degree of the given polynomial~$f\in\mathbb{Z}[x]$, and approximate a random conjugate of~$f(\zeta_n)$. 

Since the algorithm is insensitive to the choice of~$n$, 
as long as it is larger than the degree of~$f$ (that is at most~$2^s$ where~$s$ is the size of circuit), we  choose~$n=2^{4s}$ ensuring that 
$\zeta^a_n$ is a conjugate of $\zeta_n$ for all odd numbers~$a$, $1\leq a<n$. 
This prevents one type of error in our randomised algorithm for the {\skct} 
problem (the error caused by picking a non-conjugate in Line~$1$ of 
Figure~\ref{fig:algorithmSkew}); indeed, 
whenever $f(\zeta_n)=0$  our algorithm   returns `Zero' with probability~$1$.
Then we conclude the following theorem noting that the  
approximation is efficiently computable in randomized sequential 
time  using Brent's algorithm~\cite{brent}. 


\begin{restatable}{theorem}{slpserial} 
\label{theo:slpserial}
Testing equality of two compressed words, of combined size~$s$,   
 is solvable  in $\widetilde{O}(s^2)$-time by a randomized 
 sequential algorithm. Furthermore,  it can be 
 implemented by $\widetilde{O}(s^3)$-sized {\NC} circuits 
 using $O(s)$ random bits.
\end{restatable}

\begin{proof}

For a randomized sequential algorithm in time $\widetilde{O}(s^2)$, having chosen the random conjugate~$\zeta_n^a$, 
for each  $x^m$, inputted to a multiplication gate, we need to
   compute 
  $\zeta^{am}_n$ truncated up to an $O(s)$-bit precision using Taylor expansion. By
  Brent's algorithm~\cite{brent}, for each $k$, $1 \leq k \leq n$, we can compute  $e^{\frac{2\pi i k}{n}}$ within an error of $2^{-O(s)}$ in $O(s\log{s})$ time. Since there are at most $O(s)$ such different occurrences of $\zeta_n^{am}$ in the powerful skew circuit, all these $O(s)$-bit approximations can be computed in $O(s^2\log{s})$-time.

We are now left with the task of evaluating a powerful-skew arithmetic circuit that has $O(s)$ binary additions and $O(s)$ binary multiplications on $O(s)$-bit numbers. Addition and multiplication of two $O(s)$-bit integers can be implemented in $O(s)$ and $O(s\log{s})$ time~\cite{harvey} respectively. Hence, for the whole circuit this can be implemented with an additional time complexity of $O(s^2) + O(s^2\log{s})$. Hence the overall time complexity is $\widetilde{O}(s^2)$. The number of random bits used is $O(\log{n}) = O(s)$ (to select a conjugate of $\zeta_n$). Notice that in a RAM model where each operation is unit cost, this results in a $O(s)$-time algorithm, and in the log-cost model a $O(s\log{s})$-time algorithm.

\medskip

The second item is an immediate consequence of Theorem~\ref{theo:skew} and its proof.
\end{proof}

\section{An NC Algorithm for \sct}\label{sec:nc}
Cheng \emph{et al.}~\cite{cheng-tv} showed how to solve the {\sct} in polynomial time.
Their method involves a tensor decomposition of the 
space of all polynomials that vanish on a given root of unity $\zeta_n$, 
based on a partial factorisation
of $n$.  They then exploit
sparsity to efficiently determine membership of this space (which has
dimension $n$, exponential in the length of the problem instance).  
Below we reformulate this idea  
to avoid working with vector spaces of exponential dimension.  
We work instead with a space  
of vanishing sums (see~\eqref{eq-vanish-sum}) 
whose dimension equals the number of monomials of the input polynomial.
Our reformulation relies on a simple proposition in linear algebra (Proposition~\ref{prop:span}),
which not only simplifies the approach of
Cheng {et al.}, but allows to place the problem {\sct} in~{\NC}.

Let $\zeta_n$ denote a primitive $n$-th root of unity for a positive
integer $n$.  Given a vector~$\boldsymbol{k}=(k_1,\ldots,k_s)$ of 
nonnegative integers where  $ k_1 < \cdots < k_s$,
we aim to compute the space of \emph{vanishing sums} 
\begin{equation} \label{eq-vanish-sum}
V_{n}^{\boldsymbol{k}}:=\left\{
a\in\mathbb{Q}^s :\sum_{i=1}^s a_i \zeta^{k_i}_n = 0\right\}\end{equation} 
in polylogarithmic parallel time
in the total bit length of $n$ and~$k_1,\ldots,k_s$.

\subsection{Composing spaces of vanishing sums}
 In the approach of~\cite{cheng-tv}[Section 2.1] the following (which is an easy
consequence of the Chinese Remainder Theorem) plays a
central role:
\begin{proposition}
Suppose that $n=n_1n_2$ for $n,n_1,n_2\in \mathbb{N}$, with $n_1$
and $n_2$ coprime.  Then the map $\zeta_n \mapsto \zeta_{n_1} \otimes
\zeta_{n_2}$ defines a $\mathbb{Q}$-algebra isomorphism between $\mathbb{Q}(\zeta_n)$
and $\mathbb{Q}(\zeta_{n_1}) \otimes \mathbb{Q}(\zeta_{n_2})$.
\label{prop:iso}
\end{proposition}

Given vectors $a,b\in\mathbb{Q}^s$, define the \emph{Hadamard product}
$a\odot b \in \mathbb{Q}^s$ by $a\odot b := (a_1b_1,\ldots,a_sb_s)$.
Furthermore, given $U,V$ vector sub-spaces of $\mathbb{Q}^s$, define
\[ U \odot V = \mathrm{span}\{ u \odot v  : u \in U, v\in V\} \, . \]
Given bases of $U$ and $V$, we can compute a basis of $U \odot V$ in
$\NC$ by constructing the set of products $u\odot v$ as $u$ ranges over 
the basis of $U$ and $v$ ranges over the basis of $V$, and then selecting
a maximally linearly independent subset of the resulting collection of vectors
(e.g., by the algorithm of~\cite{ChistovG84}).  
The operator $\odot$ is moreover associative, so given a list of 
subspaces $U_1,\ldots,U_\ell \subseteq \mathbb{Q}^s$, we can compute 
the iterated product $U_1\odot \cdots \odot U_\ell$ in $\NC$ using the parallel 
prefix technique~\cite{LadnerF80}.

Denote by
$U^\perp$ the orthogonal complement of~$U\subseteq \mathbb{Q}^s$.
\begin{proposition}
Let $U,V$ be finite dimensional vector spaces over $\mathbb{Q}$ with $u_1,\ldots,u_s \in U$
and $v_1,\ldots,v_s \in V$ for some $s\in\mathbb{N}$.  
Define the following three vector subspaces of $\mathbb{Q}^s$:
\begin{eqnarray*} 
A&:=&\textstyle \left\{ a \in \mathbb{Q}^s : \sum_{i=1}^s a_i u_i = 0\right\} \\
B&:=&\textstyle \left\{ b \in \mathbb{Q}^s : \sum_{i=1}^s b_i v_i = 0\right\} \\
C&:=&\textstyle \left\{ c \in \mathbb{Q}^s : \sum_{i=1}^s c_i (u_i\otimes v_i) = 0\right\} \, .
\end{eqnarray*}
Then $C^\perp = A^\perp \odot B^\perp$. 
\label{prop:span}
\end{proposition}
\begin{proof}
For a non-negative integer $k$ and list of vectors
$w_1,\ldots,w_s \in \mathbb{Q}^{k}$, write $R(w_1,\ldots,w_s)$ for the
row space of the matrix with columns $w_1,\ldots,w_s$.  Recall that
$R(w_1,\ldots,w_s)$ is the orthogonal complement of $\{ a \in                                                                                                                        \mathbb{Q}^s : \sum_{i=1}^s a_iw_i=0\}$.

Without loss of generality, suppose that $U=\mathbb{Q}^m$ and
$V=\mathbb{Q}^n$.  Then we can identify $U\otimes V$ with
$\mathbb{Q}^{mn}$ by taking  $u\otimes v$ to be the Kronecker product of $u\in U$ and $v\in V$. Now we have
\begin{equation}
\begin{array}{rcl}
A^\perp &\,=\,& R(u_1,\ldots,u_s) \\
B^\perp &=& R(v_1,\ldots,v_s) \\
C^\perp &=& R(u_1\otimes v_1,\ldots,u_s\otimes v_s) \, . 
\end{array}
\label{eq:perp1}
\end{equation}

But it clearly also holds that
\begin{align}
R&(u_1\otimes v_1,\ldots,u_s\otimes v_s) = 
R(u_1,\ldots,u_s) \odot R(v_1,\ldots,v_s) \, .  
\label{eq:perp2}
\end{align}

The result follows from Equations (\ref{eq:perp1}) and (\ref{eq:perp2}).
\end{proof}

The following result follows from 
Propositions~\ref{prop:iso}
and~\ref{prop:span}:
\begin{corollary}
Let $n_1$ and $n_2$ be coprime positive integers, then
\[ (V_{n_1n_2}^{\boldsymbol{k}})^\perp =
    (V_{n_1}^{\boldsymbol{k}})^\perp \odot 
    (V_{n_2}^{\boldsymbol{k}})^\perp \, .\]
\label{corl:split}
\end{corollary}

\subsection{Base cases}
We will use Corollary~\ref{corl:split} in tandem with the following known 
characterisations of the vanishing spaces for prime powers and composite numbers. 


\begin{proposition}
Let $p$ be a prime, $e$ a positive integer, and let
$0 \leq k_1<\ldots<k_s<p^{e}$ be non-negative integers.  Given
$a \in \mathbb{Q}^s$, 
we have 
$\sum_{i=1}^s a_i \zeta_{p^e}^{k_i} = 0$ if and only if
(i)~$a_i=a_j$ for all $i,j$ such that ${k_i \equiv k_j}\pmod {p^{e-1}}$
and (ii)~$a_i=0$ for all $i$ such that $\#\{ k_j : {k_i \equiv k_j}\pmod {p^{e-1}}\} < p$.
\label{prop:char}
\end{proposition}
\begin{proof}
Recall that the minimal polynomial of $\zeta_{p^e}$ is
\[ f(x)=1+x^{p^{e-1}} +x^{2p^{e-1}} + \ldots + x^{(p-1)p^{e-1}} \, .\]
For $a \in
\mathbb{Q}^s$ we have $\sum_{i=1}^s a_i \zeta_{p^e}^{k_i} = 0$ if and
only if there exists $q \in \mathbb{Q}[x]$, $\mathrm{deg}(q)<p^{e-1}$,
such that 
\[
\sum_{i=1}^s a_i x^{k_i} = q(x)f(x)
                        = \sum_{i=0}^{p-1} q(x) x^{i(p^{e-1})} \, .
\]
In other words, the polynomial $\sum_{i=1}^s a_i x^{k_i}$ consists of
$p$ appropriately translated copies of $q(x)$.  The result 
follows.
\end{proof}

\begin{proposition}
Let $f(x)=\sum_{i=1}^s a_i x^{k_i} \in \mathbb{Q}[x]$ be a polynomial
such that $0\leq k_1<\cdots<k_s<n$ and suppose that $p>s$ for  all prime divisors
$p$ of~$n$.  Then $f(\zeta_n)=0$ only if $f$ is identically zero.
\label{prop:no-small}
\end{proposition}
\begin{proof}
Write $n=p_1^{e_1}\cdots p_m^{e_m}$ for the prime factorization of
$n$.  Write $\ell_{ij}:= k_i \bmod {p_j^{e_j}}$ for $i=1,\ldots,s$ and
  $j=1,\ldots,m$.  By the Chinese Remainder Theorem the $m$-tuples
  $\ell_i=(\ell_{i1},\ldots,\ell_{im})$, $i=1,\ldots,s$, are all distinct.
Now we have
\begin{eqnarray*}
f(\zeta_n)=0 & \Leftrightarrow & \sum_{i=1}^s a_i \zeta_n^{k_i}=0\\
             & \Leftrightarrow & \sum_{i=1}^s a_i (\zeta_{p_1^{e_1}}^{\ell_{i1}}
                           \otimes \cdots \otimes \zeta_{p_m^{e_m}}^{\ell_{im}}) = 0 \, .
\end{eqnarray*}
But, by Proposition~\ref{prop:char}, $\left\{
\zeta_{p_j^{e_j}}^{\ell_{1j}},\ldots,\zeta_{p_j^{e_j}}^{\ell_{sj}}\right\}$ is
a linearly independent set in $\mathbb{Q}(\zeta_{p_j^{e_j}})$ 
for all $j=1,\ldots,m$ (possibly listed with repetitions).  It follows that
\[ \left\{ \zeta_{p_1^{e_1}}^{\ell_{i1}} \otimes \cdots \otimes \zeta_{p_m^{e_m}}^{\ell_{im}} : i=1,\ldots,s\right\} \]
is a linearly independent set in $\mathbb{Q}(\zeta_n)$.
Since the $\ell_i$ are all distinct we conclude that $a_1=\cdots=a_s=0$.
\end{proof}

\subsection{Putting Things Together}
\label{sec:new}



\theosparse*
\begin{proof}
Given $f(x) =\sum_{i=0}^s a_i x^{k_i}$ and $n\in\mathbb{N}$, we wish
to determine whether $f(\zeta_n)=0$.  

Since integer division is in {\NC}, given $n\in \mathbb{N}$ one can compute
in {\NC} a factorisation $n=p_1^{e_1} \cdots p_\ell^{e_\ell} m$ such that all
$p_1,\ldots,p_{\ell} \leq s$ are prime and all prime factors of $m$ are
strictly greater than~$s$.

Let $\boldsymbol{k}=(k_1,\ldots,k_s)$.
We use Propositions~\ref{prop:char} and~\ref{prop:no-small} to compute 
the vanishing spaces $V_{p_i^{e_i}}^{\boldsymbol{k}}$ for
$i=1,\ldots,\ell$ and $V_m^{\boldsymbol{k}}$.  More precisely,
to compute $V_m^{\boldsymbol{k}}$, we let $0 \leq k'_1<\cdots < k'_t < m$
be a list of the distinct residues of $k_1,\ldots,k_s$ modulo $m$, and define 
a map $T:\mathbb{Q}^s\rightarrow \mathbb{Q}^t$ by $T(a_1,\ldots,a_s)=(b_1,\ldots,b_t)$,
where $b_i := \sum \{ a_j : k_j \equiv k'_i \bmod m\}$ for $i=1,\ldots,\ell$.  Then
$V_m^{\boldsymbol{k}}$ is the pre-image under $T$ of $V_m^{(k'_1,\ldots,k'_t)}$. 
Sine all prime factors of $m$ are greater than $s$, by Proposition~\ref{prop:no-small}, $V_m^{\boldsymbol{k}}$ is the preimage of~$T(\boldsymbol{0})$.
The computation
of $V_{p_i^{e_i}}^{\boldsymbol{k}}$ is analogous and uses Proposition~\ref{prop:char}.
Moreover, since only
integer division is required to specify the linear map $T$, the given characterisations can be computed in {\NC}~\cite{bch}.
Finally, the orthogonal complements of the above-computed vanishing spaces can also be derived 
in~{\NC}~\cite{mulmuley}.  

By Corollary~\ref{corl:split} we have that 
\begin{align*} (V_n^{\boldsymbol{k}})^\perp = \; 
   (V_{p_1^{e_1}}^{\boldsymbol{k}})^\perp \odot \cdots
   \odot 
   (V_{p_\ell^{e_\ell}}^{\boldsymbol{k}})^\perp  \odot (V_m^{\boldsymbol{k}})^\perp \, . 
   \end{align*}
But, as observed above, such an iterated product can be computed in
$\NC$.  Finally, with 
$(V_n^{\boldsymbol{k}})^\perp$ in hand we can directly test whether
$f(\zeta_n)=0$.
\end{proof}

We use the above compositional  technique to recover the  result of Migotti~\cite{migotti} (see also Bang~\cite{BANG}): 

\begin{restatable}{proposition}{cyclpoly}
For all $n\in \mathbb{N}$ with at most two odd prime divisors,
all coefficients of the $n$-th cyclotomic polynomial lie in $\{-1,0,1\}$.
\end{restatable}
\begin{proof}
Let $n=p^{k_1}q^{k_2}2^{k_3}$ be an integer with two odd prime divisors~$p$ and $q$.
Given two integers $m,k$, denote by $V_m$  the vanishing
space $\{ a \in \mathbb{Q}^{\varphi(m)+1}: 
\sum_{i=0}^{\varphi(m)} a_i\zeta_m^i=0\}$.
Let~$\mathbb{1}_{m} \in \mathbb{Q}^m$ be the column-vector  with all entries equal to $1$, 
and ~$I_{m}\in \mathbb{Q}^{m\times m}$ be the identity matrix. 
Let $e_{m,k}\in \mathbb{N}^{m^{k-1}}$ be the row vector whose last coordinate is~$1$, with
all other coordinates zero.

Proposition~\ref{prop:char} characterizes the vanishing space  $V_m$ for primes~$m$ as the null space 
of the matrix~$\begin{bmatrix}\mathbb{1}_{m-1} &-I_{m-1}\end{bmatrix}$.
It also states that $V_{m^k}$ for prime powers~$m^k$ 
is the null space 
of $N(m,k):=\begin{bmatrix}\mathbb{1}_{m-1} &-I_{m-1}\otimes e_{m,k}\end{bmatrix} $.

By Proposition~\ref{prop:iso}, 
the vanishing space $V_n$ is the set of solutions of 
\begin{equation} \label{eq:CD}
C \cdot  \begin{bmatrix}
a_0  & \cdots & a_{\varphi(n)} 
\end{bmatrix}^T= 0. 
\end{equation}
where $C$ is the $\varphi(n)\times (\varphi(n)+1)$ matrix 
whose $i$-th column is the tensor product of  column
$i \, {\bmod} \, p^{k_1}$ in $N(p,k_1)$, column
$i \, {\bmod} \, q^{k_2}$ in $N(q,k_2)$ and column
$i \, {\bmod} \, 2^{k_3}$ in~$N(2,k_3)$. Since
$C$ is a submatrix of 
$N(p,k_1)\otimes N(q,k_2)\otimes N(2,k_3)$, by  
 Proposition~\ref{prp:tupq} in Appendix~\ref{app:totuni}, 
it is totally unimodular. 

Since $\zeta_n$ has norm one,
the constant term of $\Phi_n(x)$ is~$1$. 
We then search for the unique solution of~\eqref{eq:CD} 
such that  $a_0=1$. Since all entries in the  
zero-th column of $C$ are~$1$,  the  submatrix~$D$ of~$C$ 
obtained by deleting this column is such that
\begin{equation*} 
D\cdot \begin{bmatrix}
 a_1 & \cdots & a_{\varphi(n)} 
\end{bmatrix}^T= -\mathbb{1}_{\varphi(n)}
\end{equation*}
Let $D_i$ be obtained from~$D$ by replacing the $i$-th column 
with $-\mathbb{1}_{\varphi(n)}$.
Then both $det(D)$ and $det(D_i)$ are minors of $C$. 
 Applying Cramer's rule,  
    $a_i=\frac{det(D_i)}{det(D)}\in \{-1,0,1\}$.

\end{proof}

\section{Diagonal Circuits}\label{sec:diagonal}
\newcommand{\x}{{\bf x}}
\newcommand{\y}{{\bf y}}
\newcommand{\ppoly}{{\sf P/poly}}

In this section we study CIT for the class of diagonal
circuits~\cite{saxena-diagonal}.  In the multivariate 
formulation of CIT these compute polynomials 
of the form $f=\sum_{i=1}^s g_i^{d_i}$, where the $g_i$ are linear forms
and the $d_i$ are integers in unary.
The resulting CIT problem is a special case of {\skct} and generalises 
{\sct}.  
Note that in our univariate formulation the $g_i$ become 
polynomials in sparse representation.
We start with the following hardness result.

\begin{theorem}
\label{thm:shiftedsct}
If CIT for diagonal circuits is solvable in polynomial time then
PIT for algebraic circuits of size $s$ and degree $d \leq s$ can be 
solved in $s^{O(\sqrt{d})}$ time. 
\end{theorem}
\begin{proof}
It is known
(see for example~\cite[Theorem~5.17]{saptharishi-survey} 
or~\cite[Proposition~1]{kkps}) that an $m$-variate polynomial 
$P(\x)$ of degree at most $d = m^{O(1)}$ computed by an 
algebraic circuit of size $s = m^{O(1)}$  
can be expressed as 
$P(\x) = \sum_{i=1}^{s^{O(\sqrt{d})}} c_i Q_i(\x)^{d_i}$,
where the $Q_i$ have $s^{O(\sqrt{d})}$ monomials and degree $O(\sqrt{d})$.
Moreover such a representation can be computed $s^{O(\sqrt{d})}$ time.
As a result, a deterministic $\poly(s, d, m)$-time algorithm for identity testing polynomials of the above form will yield a $s^{O(\sqrt{d})}$-time algorithm for identity testing arbitrary algebraic circuits. Applying Kronecker substitution~\cite{kronecker} on $P(\x)$ expressed in the above form yields a univariate polynomial $p(x)=\sum_{i=1}^s c_i (a_{i1}x^{e_{i1}} + \dots + a_{is}x^{e_{is}})^{d_i}$
such that $P(\x) \neq 0$ if and only if $p(x) \neq 0$. 
Note that for all~$n$ such that $\varphi(n)> d(d+1)^m$, the degree of $p$ is smaller  than the $n$-the cyclotomic polynomial. Thus $p(x) \neq 0$ if and only if~$p(\zeta_n) \neq 0$.  
\end{proof}

In Appendix~\ref{sec:saxena}, Proposition~\ref{prop:shift}
builds on the proof of Theorem~\ref{thm:shiftedsct} to show 
that efficient algorithms for another simple variant of sparse-CIT, namely evaluating low-degree sparse multivariate polynomials $f(x_1, \dots, x_k)$ at \emph{translations} of roots of unity, will yield a sub-exponential-time algorithm for PIT. More formally, if evaluating $f(a_1 + \zeta_n^{e_1}, \dots, a_k + \zeta_k^{e_k})$ where $f$ is a $k$-variate, $k^{O(1)}$-sparse polynomial and $a_i \in \mathbb{Q}$, $e_i \in \mathbb{N}$ are specified in binary is decidable in $\poly(k, \log{n})$ time, then PIT for algebraic circuits of size $s$ and degree $d$ can be decided in $s^{O(\sqrt{d})}$-time.

We now specialise to consider the subclass of diagonal circuits that compute
polynomials of the form $f=\sum_{i=1}^s g^{d_i}$ with $g$ a single 
univariate polynomial in sparse representation.
We give an algorithm that solves the CIT for this class of circuits in polynomial time 
assuming GRH.

\begin{figure}[t]
\begin{center}
\begin{tabular}{rp{.9\textwidth}}
 \hline
    \multicolumn{2}{c}{ \bf  Restricted Diagonal Circuits }  \\
    \hline
 & {\bf Input:} Polynomial $f=\sum_{i=1}^s g^{d_i}$, 
 with $g$ in sparse representation and $d_i$ in unary
\\\hline
& {\bf Output:} Whether $f(\zeta_n)=0$ for $n$ written in binary.\\\hline
  1:& Set $d:=\max_{i=1}^s d_i$.\\
  2:& Compute the orbit $\mathrm{Orb}(g(\zeta_n))$ of $g(\zeta_n)$ w.r.t.\ the set $\{ k \in \Zn^* : k \leq G(n) \}$.\\
  3:& If $|\mathrm{Orb}(g(\zeta_n))| > d$ then return 'Non-Zero'.\\
  4:& If $|\mathrm{Orb}(g(\zeta_n))| \leq d$ then compute $\alpha \in \mathbb{Q}(i)$ such that $|\alpha-f(\zeta_n)|<\frac{\varepsilon(f)}{3}$ and 
         return `Zero' if $\alpha<\frac{\varepsilon(f)}{3}$ and return `Non-Zero' otherwise.
\\
   \hline
\end{tabular}
\caption{Algorithm for Restricted Diagonal Circuits}
\label{fig:algorithm-3}
\end{center}
\vspace{-.5cm}
\end{figure}

\theodiamond*
\begin{proof}
The algorithm is given in Figure~\ref{fig:algorithm-3}.  
It involves an integer parameter $G(n)$ and a rational parameter 
$\varepsilon(f)$ that are both functions of the input.
We will say more about both parameters shortly, suffice to say 
for now that $G(n)$ is chosen such that
$\{ k \in \Zn^* : 1 \leq k \leq G(n)\}$ generates~$\Zn^*$.   

We first argue correctness and then move to analysing the complexity.
Line~2 refers to the action 
of the group~$\Zn^*$ on the field $\mathbb{Q}(\zeta_n)$, obtained by associating with  
$k \in \Zn^*$ the automorphism  $\sigma_k$ of $\mathbb{Q}(\zeta_n)$ defined by
$\sigma_k(\zeta_n)=\zeta_n^k$.  
Specifically,
Line~2 computes the orbit
$\mathrm{Orb}(g(\zeta_n))$ of
$g(\zeta_n)$ under the subgroup 
of $\Zn^*$ generated by the set
$\{ k \in \Zn^* : k \leq G(n)\}$, that is,
the smallest set that contains
$g(\zeta_n)$ and is closed under the action of the aforementioned subgroup.

Observe that when the algorithm halts in Line~3  the output is correct:
 if  $g(\zeta_n)$ has more than~$d$ distinct conjugates then it cannot be that
 $f(\zeta_n)=\sum_{i=1}^s g(\zeta_n)^{d_i}=0$.

Now suppose that $|\mathrm{Orb}(g(\zeta_n))| \leq d$ in Line~3.  
We will use this assumption to bound the degree and height of $f(\zeta_n)$.
By the assumption that $\{ k \in \Zn^* : 1 \leq k \leq G(n)\}$ generates $\Zn^*$,
we have that $\mathrm{Orb}(g(\zeta_n))$ consists of all Galois conjugates of $g(\zeta_n)$.  Since
$|\mathrm{Orb}(g(\zeta_n))| \leq d$ it follows that
$g(\zeta_n)$, and hence also~$f(\zeta_n)$, have degree at most $d$.
Furthermore, for every $\ell\in \Zn^*$ we have $|f(\zeta^\ell_n)| \leq sM^d$, where 
$M$ is the sum of the absolute value of all coefficients of $g$.
By writing the coefficients 
of the minimal polynomial of $f(\zeta_n)$ in terms of the Galois conjugates of $f(\zeta_n)$,
we have that $f(\zeta_n)$ has height at most $H:=2^d sM^d$.  

But a non-zero algebraic number of degree $d$ and height~$H$ has magnitude at 
least $\frac{2}{d^{d+1}H^d}$~\cite{Mignotte1983}.  Thus if we define
\begin{gather}
    \varepsilon(f):=\frac{2}{d^{d+1}H^d} \, ,
    \label{eq:epsilon}
\end{gather} we have that
if $f(\zeta_n)\neq 0$ then $|f(\zeta_n)|>\varepsilon(f)$: hence
for the number $\alpha$ computed in Line~4 we have
$|\alpha|>\frac{2\varepsilon(f)}{3}$.  On the other hand,
if $f(\zeta_n)=0$ then $|\alpha|<\frac{\varepsilon(f)}{3}$.
Thus the output produced in Line~4 is correct.
This completes the proof that the algorithm gives the correct output.  

We turn now to the complexity.
Note that we can use the procedure presented in the previous section
to determine in polynomial time the equality of two 
conjugates $g(\zeta_n^\ell)$ and~$g(\zeta_n^j)$ 
of $g(\zeta_n)$.
Since the computation of $\mathrm{Orb}(g(\zeta_n))$ terminates 
as soon as $|\mathrm{Orb}(g(\zeta_n))|>d$, we see that Line~2 can 
be executed in time polynomial in the size of the input and the parameter $G(n)$.
Now it was shown in~\cite{Montgomery71} that under GRH there is a 
function $G(n)=O(\log^2 n)$ such that 
$\{ k \in \Zn^* : 1 \leq k \leq G(n)\}$ generates 
$\Zn^*$.\footnote{The paper~\cite{Bach93} gives 
heuristic arguments and experimental data suggesting that the choice $G(n)=(\ln 2)^{-1} \ln n \ln \ln n$} will yield 
a set of generators.
It follows that Line 2 of the procedure can be executed in
polynomial time, assuming GRH.  Finally, from Expression~\eqref{eq:epsilon} we see that 
$|\log(\varepsilon(f))|$
is polynomially bounded in the input size.  
Thus $f(\zeta_n)$ can be computed to within precision~$\frac{\varepsilon(f)}{3}$ in 
polynomial time, e.g., using the approach via the Taylor expansion as described in Section~\ref{sec:RandNC}.
\end{proof}

\bibliographystyle{plain}
\bibliography{refer}

\newpage 

\appendix

\section{Missing Proofs in Section~\ref{sec:nc}}\label{app:totuni}

Recall that~\cite[Chapter~19]{TUref} given a $m\times n$ matrix~$A$ where all entries are in $\{1,0,-1\}$,
$A$ is totally unimodular (TU) if and only if for all subsets $C \subseteq \{1,\cdots, n\}$ of columns,
there exists a coloring of $C$ with $B$ and $G$ such that for all rows $i\in \{1,\ldots, m\}$:
\begin{equation}\label{eq:conditiontu}
\sum_{j\in B} a_{i,j}-\sum_{j\in G} a_{i,j}\in \{0,1,-1\}.
\end{equation}
Given a  matrix~$A$, we say that $A$ is nonnegative TU if
there exists a coloring that in addition to~\eqref{eq:conditiontu} for all rows $i\in \{1,\ldots, m\}$ it satisfies 
$\sum_{j\in B} a_{i,j}-\sum_{j\in G} a_{i,j}\in \{0,1\}$; we say that  
$A$ is constant TU if it is nonnegative and for all rows the above sum takes the same value.

Let $e$ be a vector whose coordinates are all zero, except one that equals $1$. 
For all integers~$m$ and all
constant-TU matrix~$M$ the followings hold:
\begin{itemize}
\item $M \otimes e$ is constant TU,
\item $e \otimes M$ is constant TU, 
\item $\begin{bmatrix}\mathbb{1}_{m} &-I_{m}\end{bmatrix}$ is nonnegative TU,
\item $\begin{bmatrix}\mathbb{1}_{m} &-I_{m}\end{bmatrix}\otimes M $ is nonegative TU. 
\end{itemize}

\begin{proposition}\label{prp:tupq}
Let  $p,q$ be primes and $e_{k}\in \mathbb{N}^k$ a  row vector whose coordinates are all zero, except last one that equals $1$. For
all numbers~$k_1,k_2,k_3 \in \mathbb{N}$,
\[\begin{bmatrix}\mathbb{1}_{p-1} &-I_{p-1} \otimes e_{k_1} \end{bmatrix} \otimes \begin{bmatrix}\mathbb{1}_{q-1} &-I_{q-1} \otimes e_{k_2}\end{bmatrix}  \otimes \begin{bmatrix}\mathbb 1 &-1 \otimes e_{k_3}\end{bmatrix}  \]
is totally unimodular.
\end{proposition}
\begin{proof}
We prove the totally unimodularity of a matrix $N$ defined by
\[\begin{bmatrix}\mathbb{1}_{p-1} &-I_{p-1}  \end{bmatrix} \otimes e_{k_1} \otimes \begin{bmatrix}\mathbb{1}_{q-1} &-I_{q-1} \end{bmatrix}  \otimes e_{k_2} \otimes \begin{bmatrix}\mathbb 1 &-1 \end{bmatrix} \otimes e_{k_3}.\] By totally unimodualrity of all submatrices of a TU matrix the result follows. 

Since  $\begin{bmatrix}\mathbb 1 &-1\end{bmatrix}$ is
constant TU, then $ e_{k_2} \otimes \begin{bmatrix}\mathbb 1 &-1 \end{bmatrix} \otimes e_{k_3} $
is constant TU.
Since $\begin{bmatrix}\mathbb{1}_{q-1} &-I_{q-1}\end{bmatrix}$ is nonnegative TU, we get that 
$M:=e_{k_1} \otimes \begin{bmatrix}\mathbb{1}_{q-1} &-I_{q-1}\end{bmatrix} \otimes e_{k_2} \otimes \begin{bmatrix}\mathbb 1 &-1\end{bmatrix} \otimes e_{k_3}$ is nonngative TU.

Given a set~$C\subseteq \{1,\ldots, n\}$ of columns for $N$,  we  colour columns in $C$ with $G$ and $B$
by viewing  $N$  as 
\[N:=\begin{bmatrix} M&-M & 0 &\ldots &0\\
M &0 & -M &\ldots &0\\
\vdots &\vdots & \vdots &\vdots &\vdots\\
M &0 & 0 &\ldots &-M\\
\end{bmatrix}.\]
Define
\[C_i=C \cap \{j+(i-1)p\mid 1\leq j\leq q\}\]
Color all columns in $C_1$, that is $C$ restricted to first $q$ columns of $N$, with $G$.
Since $M$ is nonnegative TU,  
 the sum  of  entries of  columns $C_1$ is either in $\{1,0\}$ for all rows, or in $\{0,-1\}$ for all rows.
 
This property holds for all  $i=2,\ldots, p$.
For each $C_i$, if the sum of entries of $C_i$ 
agrees with the sum of $C_1$, color $C_i$ with~$B$ 
otherwise color it with $G$. 

Observe that for all rows the only nonnegative entries appear in $C_1$ and a single $C_i$, $i\in \{2,\ldots,p\}$. Hence, the above colouring satisfy~\eqref{eq:conditiontu}.
\end{proof}

\section{Missing Proofs from Section~\ref{sec:diagonal}} \label{sec:saxena}
The following \emph{duality lemma} due to Saxena~\cite{saxena-diagonal} expresses power of a multivariate linear form as a linear combination of product of univariate polynomials:

\begin{restatable}{proposition}{propduality}
	For every $m, d > 0$, there exist $\alpha_i, \beta_{ij} \in \Q$ ($0 \leq i \leq md, 0 \leq j \leq d$) such that
	
	\begin{align*}
	(a_1u_1 + \dots + a_mu_m)^d &= \sum_{i=0}^{md}\sum_{j=0}^d \beta_{ij} (u_1 + \alpha_{i1})^j \dots (u_m + \alpha_{im})^j.
	\end{align*}
\label{prop:duality}
\end{restatable}

We provide a proof due to Gupta et al.~\cite{gkks}
\begin{proof}
	Consider 
	\begin{align*}
	P_u(z) &= (z+a_1u_1) \dots (z+a_mu_m) - z^m \\
	&= z^{m-1} (a_1u_1 + \dots a_mu_m)^d + \mbox{lower order terms} \\
	 P_u(z)^d &= z^{(m-1)d}(a_1u_1 + \dots + a_mu_m)^d + \mbox{lower order terms} \\
	\end{align*}
	
Hence we can compute $(a_1u_1 + \dots + a_mu_m)^d$ as a coefficient of $z^{(m-1)d}$ via interpolation by evaluating $P_u(z)^d$ on $(m-1)d$ points. That is, for every distinct $\alpha_0, \dots, \alpha_{md} \in \Q$, there exist $\beta_{0}' \dots \beta_{md}'$ such that
	\begin{align*}
	(u_1 + & \dots  + u_m)^d \\ &=\sum_{i=0}^{md} \beta_i' P_u(z)^d \\
	&= \sum_{i=0}^{md} \beta_i'((a_1u_1 + \alpha_i) \dots (a_m u_m + \alpha_i) - \alpha_i^m)^d \\
	&= \sum_{i=0}^{md} \beta_i' \sum_{j=0}^d \binom{d}{j} (-\alpha_i^m)^{(d-j)}((a_1u_1 + \alpha_i) \dots (a_mu_m + \alpha_i))^j \\
	&= \sum_{i=0}^{md}\sum_{j=0}^d \beta_{ij} ((u_1 + \alpha_{i1}) \dots (u_m + \alpha_{im}))^j \\
	\end{align*}
	
	where $\beta_{ij} = \beta_i' \binom{d}{j} (-\alpha_i)^{m(d-j)} (a_1 \dots a_m)^j$ and $\alpha_{i1} = \frac{\alpha_i}{a_1}$, ..., $\alpha_{im} = \frac{\alpha_i}{a_m}$.
\end{proof}

\begin{proposition}\label{prop:shift}
Given a $k$-variate sparse polynomial $f$ and numbers
$a_1,\ldots,a_k \in \mathbb{Q}$ and  $e_1,\ldots,e_k \in \mathbb{N}$ 
in binary, if one can test $f(a_1 + \zeta_n^{e_1},\ldots, a_k + \zeta_n^{e_k}) = 0$ 
in deterministic $\poly(k, \log{n})$ time, then  
PIT for $m$-variate polynomials computed by algebraic circuits of size $s$ and degree $d \leq s = m^{O(1)}$ can be 
solved in $s^{O(\sqrt{d})}$ time. 
\end{proposition}

\begin{proof}
From Theorem~\ref{thm:shiftedsct}, identity testing polynomials of the form $p(x)=\sum_{i=1}^s c_i (a_{i1}x^{e_{i1}} + \dots + a_{is}x^{e_{is}})^{d_i}$ in deterministic $\poly(s,m,d)$ time suffices to get a $s^{O(\sqrt{d})}$ time algorithm for PIT.
Using 
Proposition~\ref{prop:duality} to simplify the 
terms $(a_{i1}x^{e_{i1}} + \dots + a_{is}x^{e_{is}})^{d_i}$
we can rewrite $p(x)$ as
\begin{align*}
p(x) &= \sum_{i=1}^s c_i \sum_{r=0}^{s d_i}\sum_{j=0}^{d_i} \beta_{irj} (x^{e_{i1}} + a_{ir1})^j \dots (x^{e_{is}} + a_{irs})^j,
\end{align*}
which is a univariate polynomial of degree $\leq d(d+1)^m$. 

Writing \[f = \sum_{i=1}^s c_i \sum_{r=0}^{s d_i}\sum_{j=0}^{d_i} \beta_{irj} (x_{ir1}\dots x_{irs})^j.\]
we have that~$f$ is a multivariate polynomial with at most $s^2(sd+1)$ variables and at most 
$s(d+1)(sd+1)$ monomials. Note that $p(\zeta_n)$ is the same as $f$ evaluated at $x_{ir1} = \zeta_n^{e_{i1}}+ a_{ir1}, \dots, x_{irs} = \zeta_n^{e_{irs}} + a_{irs}$. 
Furthermore, for all $n$ such that 
$\varphi(n)> d(d+1)^m$, the degree of $p$ 
is smaller than the $n$-th cyclotomic polynomial.
Thus $p(x) \neq 0$ if and only if $p(\zeta_n)\neq 0$.
\end{proof}

\end{document}